\documentclass{entcs} 

\usepackage{prentcsmacro}

\usepackage{amsmath,amsgen,latexsym}
\usepackage{amstext,amssymb,amsfonts}
\usepackage{cite}
\usepackage{multicol}
\usepackage{mathtools}
\usepackage{subcaption}
\usepackage{mathrsfs}

\usepackage{hyperref}
\usepackage{color}
\usepackage{graphicx}

\usepackage[boxed]{algorithm}
\usepackage{algpseudocode}

\newcommand{\inout}[3]{ 
	\vspace{0.2cm}
	\noindent
	\begin{tabular}{p{.13\textwidth} p{.006\textwidth} p{.82\textwidth}}
		\multicolumn{3}{l}{\scalebox{.82}{\textsc{#1}}}\\
		\hline
		\multicolumn{1}{|l}{INPUT} &:& \multicolumn{1}{p{.82\textwidth}|}{#2} \\
		\multicolumn{1}{|l}{OUTPUT} &:& \multicolumn{1}{p{.82\textwidth}|}{#3} \\
		\hline
	\end{tabular}
	\vspace{0.2cm}
}

\graphicspath{ {./figs/} }

\newcommand{\trominodr}{\begin{gathered}\includegraphics[scale=.75]{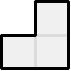}\end{gathered}} 

\def\lastname{J.T. Akagi, E.A. Canale, M. Villagra}

\begin{document}

\begin{frontmatter}
    \title{Tromino Tilings with Pegs via Flow Networks}
    \author{Javier T. Akagi\thanksref{email-tadashi}}
    \address{Facultad Polit\'ecnica, Universidad Nacional de Asunci\'on\\ NIDTEC, Campus Universitario, San Lorenzo C.P. 111421, Paraguay}
    \author{Eduardo A. Canale\thanksref{email-eduardo}}
    \address{Facultad de Ingenier\'ia, Universidad de la Rep\'ublica\\ Montevideo PC 11300, Uruguay}
    \author{Marcos Villagra\thanksref{email-marcos}}
    \address{Facultad de Ciencias Exactas y Naturales, Universidad Nacional de Asunci\'on\\ Dpto. de Matem\'aticas, Campus Universitario, San Lorenzo C.P. 111421, Paraguay}
    \thanks[email-tadashi]{Email:\href{mailto:akagi.tada@gmail.com} {\texttt{\normalshape akagi.tada@gmail.com}}}
    \thanks[email-eduardo]{Email:\href{mailto:canale@fing.edu.uy} {\texttt{\normalshape canale@fing.edu.uy}}}
    \thanks[email-marcos]{Email:\href{mailto:mvillagra@pol.una.py} {\texttt{\normalshape mvillagra@pol.una.py}}}
\begin{abstract} 
  A tromino tiling problem is a packing puzzle where we are given a region of connected lattice squares and we want to decide whether there exists a tiling of the region using trominoes with the shape of an L. In this work we study a slight variation of the tromino tiling problem where some positions of the region have pegs and each tromino comes with a hole that can only be placed on top of the pegs. We present a characterization of this tiling problem with pegs using flow networks and show that (i) there exists a linear-time parsimonious reduction to the maximum-flow problem, and (ii) counting the number of such tilings can be done in linear-time. The proofs of both results contain algorithms that can then be used to decide the tiling of a region with pegs in $O(n)$ time.
\end{abstract}
\end{frontmatter}

\section{Introduction}\label{sec:intro}

\subsection{Background}

A \emph{polyomino tiling problem} is a packing puzzle where a player attempts to cover a board or region with polyominoes of a given shape. A \emph{polyomino} is a set of square cells joined together by their edges. Of particular interest are polyominoes of two cells called \emph{dominoes}, and polyominoes of three cells called \emph{trominoes} of which there are two types: L-trominoes with the shape of an L and I-trominoes with the shape of an I. From now on, and for the remainder of this work, we will refer to L-trominoes simply as trominoes.

One of the first rigorous studies of polyominoes is due to Golomb \cite{Gol66} who introduced a hierarchy of tiling capabilities by polyominoes. Later, Conway and Lagarias \cite{CL90} presented an algebraic necessary condition to cover a region with polyominoes. In particular, the study of domino tilings uncovered many tight connections with matrix algebra through the theory of alternating-sign matrices and graph matchings \cite{EKL92}.

In computational complexity theory, Moore and Robson \cite{MR01} showed that deciding the tiling of a given region with trominoes is NP-complete. Demaine and Demaine \cite{DD07} showed the polynomial-time equivalence between polyomino tilings, edge-matching puzzles and jigsaw puzzles. Horiyama \emph{et al.} \cite{HIN17} strengthened the NP-completeness result of Moore and Robson and proved that counting the number of tromino packings of a region is \#P-complete. More recently, Akagi \emph{et al.} \cite{AGM20} showed that the NP-completeness of deciding a tromino tiling heavily depends on geometrical properties of the region; for example, a region with the shape of an Aztec rectangle admits a polynomial time algorithm that decides a tromino tiling, whereas the problem remains NP-complete for an Aztec rectangle with at least two holes.

\subsection{Contributions}
\setcounter{footnote}{0}
In all cited works of the previous paragraphs most of the arguments rely on a (partial) characterization of the tiling problem using a graph-theoretic idea. One such idea is the so-called \emph{dual graph}\footnote{Note that this is an abuse of the term dual graph used in graph theory. The practice is, however, standard in combinatorics of tilings.} of a region, which is a graph where each cell of the region is a vertex and there is an edge between two vertices if their corresponding cells are adjacent. There is a natural correspondence between domino tilings and perfect matchings in dual graphs. Likewise, Akagi \emph{et al.} \cite{AGM20} showed that tromino tilings correspond to independent sets in a slight variation of dual graphs called an intersection graph.

In this work we present yet another way to characterize tromino tiling problems using this time an idea of flow networks. This characterization, however, applies to a slight variation of the tromino tiling problem where a region contains pegs, see Fig. \ref{fig:region-pegs}. To cover a region with pegs we use a special type of tromino that we call \emph{p-tromino} which contains a hole in its corner cell as shown in Fig. \ref{fig:p-tromino}. In a tiling of a region with pegs, only p-trominos can be placed on top of pegs and we assume that there is an infinite supply of p-trominoes and only p-trominoes; that is, there are no regular trominoes with no holes and the hole of a p-tromino cannot be placed on top of a cell with no peg. See Fig. \ref{fig:p-tromino-region} for an example of a correct placement of a p-tromino. We thus define the \emph{P-Tromino Tiling Problem} as the problem of deciding, given a region with pegs, whether there is a tiling of the region with p-trominos. More formally, the decision problem is:

\begin{figure}[t]
    \begin{subfigure}[b]{.49\linewidth}
        \centering
        \includegraphics[scale=1.5]{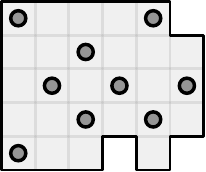}
    \end{subfigure}
    \begin{subfigure}[b]{.49\linewidth}
        \centering
        \includegraphics[scale=1.5]{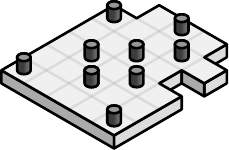}
    \end{subfigure}
    \caption{Region with pegs. Aerial view on the left and a 3D view on the right.}
    \label{fig:region-pegs}
\end{figure}

\begin{figure}[t]
    \centering
    \includegraphics[scale=2.5]{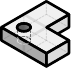}
    \caption{A p-tromino. The corner cell contains a hole and can only be placed on top of a peg on a region with pegs.}
    \label{fig:p-tromino}
\end{figure}

\begin{figure}[t]
    \begin{subfigure}[b]{.49\linewidth}
        \centering
        \includegraphics[scale=1.5]{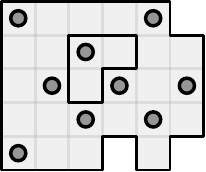}
    \end{subfigure}
    \begin{subfigure}[b]{.49\linewidth}
        \centering
        \includegraphics[scale=1.5]{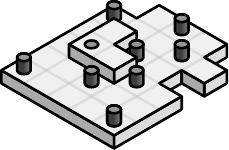}
    \end{subfigure}
    \caption{Placement of a p-tromino. Aerial view on the left and 3D view on the right.}
    \label{fig:p-tromino-region}
\end{figure}

\inout{$\textsc{P-Tromino}_{n,k}$}{
Region $R$ with $n$ cells and a list $P$ of $k$ coordinates of $R$ with pegs.
}{
``\texttt{Yes}'' if there exists a tiling of $R$ with p-trominoes; ``\texttt{No}'' otherwise.
}

Now let us recall that in a \emph{maximum-flow problem} we are given a flow network represented by a directed graph and we want to find a flow of maximum value that can be pushed from a source vertex to a sink vertex. Intuitively, the value of a flow is the rate at which some hypothetical material moves through the network. We consider the following decision problem:

\inout{$\textsc{MaxFlow}_{k}$}{
Flow network $G$, positive integer $k$.
}{
``\texttt{Yes}'' if the value of a maximum flow  in $G$ equals $k$; ``\texttt{No}'' otherwise.
}

Now we are ready to state the main result of this work.

\begin{theorem}\label{the:ptromino}
There exists a parsimonious reduction from $\textsc{P-Tromino}_{n,n/3}$ to $\textsc{MaxFlow}_{n/3}$ that is computable in linear-time.
\end{theorem}

The main idea for a proof of Theorem \ref{the:ptromino} is the construction of a flow network representation of a region with pegs such that flow passing through a path in the network, from source to sink, corresponds to the placement of a p-tromino on the region.

As a second result we show that the counting version of $\textsc{P-Tromino}_{n,n/3}$, namely $\#\textsc{P-Tromino}_{n,n/3}$, is computable in linear-time.

\begin{theorem}\label{the:counting}
$\#\textsc{P-Tromino}_{n,n/3}$ is computable in $O(n)$ time.
\end{theorem}

The idea of a proof of Theorem \ref{the:counting} exploits the structure in the flow network representation of a region with pegs. In particular, we show the existence of induced bipartite subgraphs in the flow network where any maximum matching is a perfect matching of size $n/3$ and each perfect matching agrees with a p-tromino tiling of the region. This reduces our problem of counting p-tromino tilings to counting perfect matchings in a highly structured bipartite graph.

If we are only interested in finding a p-tromino tiling, we can slightly modify our algorithm of Theorem \ref{the:counting} for counting perfect matchings and obtain an $O(n)$ time algorithm that finds a p-tromino tiling.

\subsection{Outline}
The rest of the paper is organized as follows. In Section \ref{sec:preliminaries} we introduce the notation used throughout this work and give some precise definitions of tromino tilings and flow networks. In sections \ref{sec:main-theorem} and \ref{sec:counting} we present our proofs of theorems \ref{the:ptromino} and \ref{the:counting}, respectively. Finally, we close this paper in Section \ref{sec:conclusions}.

\section{Preliminaries}\label{sec:preliminaries}
In this section we review some definitions from tromino tilings and flow networks and present the notation used throughout this work. We use $\mathbb Z$ to denote the set of all integers, and $\mathbb N$ will denote the set of all natural numbers including 0. We also denote by $[a,b]$ the discrete interval $\{a,a+1,\dots,b\}$.


\subsection{Tromino Tilings}

A region $R$ is a finite union of cells whose interior is connected. If $[a,a+1]\times [b,b+1]$ is a cell in $R$, its \emph{coordinate} in $R$ is $(a,b)$. To denote the cell at coordinate $(a,b)$ we use the shorthand notation $R_{(a,b)}$. Two cells are \emph{neighbors} or \emph{adjacent} if the Manhattan distance between the two cells is 1; thus, two cells in diagonal to each other are not adjacent.

A \emph{tromino} is a polyomino of 3 cells with the shape of an L, as we explained in the previous section. Given a tromino $T$, we will refer to the cell in its corner as the \emph{corner cell} of $T$ and the other two cells we will refer to them as the \emph{tips} of $T$. A \emph{cover} or \emph{tiling} of a region $R$ is a set of trominoes covering all cells of $R$ with no overlapping and each tromino is inside $R$. The size of a cover is the number of trominoes in it.

Given a region $R$ and a list of $k$ positions in $R$ denoting positions of pegs on the board, we define the \emph{tromino tiling with pegs} problem denoted $\textsc{P-Tromino}_{n,k}$, where each tromino comes with a hole in its corner cell and can be placed on the region $R$ only if that corner cell is placed on top of a position with a peg. A tromino with a hole in its corner is called a \emph{p-tromino}. A cover of $R$ with pegs is called a \emph{p-cover}, and $\textsc{P-Tromino}_{n,k}$ is the problem of deciding whether there is a p-cover or not.


\subsection{Flow Networks}

Given a graph $G = (V,E)$, let $V(G)$ and $E(G)$ denote its set of vertices and edges, respectively.  The \emph{degree} of a vertex $v$ in $G$ is the number $d_G(v)$ of vertices adjacent with $v$ in $G$.

A \emph{flow network} $G=(V,E)$ is a directed graph where each directed edge $(u,v)$ has an assigned capacity $c(u,v)\geq 0$.  If $(u,v)$ is an edge of $G$, then $(v,u)$ is not an edge of $G$. For edges not appearing in $G$ we will usually assign them a capacity of 0. Furthermore, self-loops are not allowed in flow networks. We also assign to each vertex $v\in V(G)$ a capacity $c(v)\geq 0$. Note that we overload our notation and we use $c(u,v)$ with $(u,v)\in E(G)$ to denote capacities on edges and $c(v)$ with $v\in V(G)$ to denote capacities on vertices.

Every flow network $G$ has two distinguished vertices called a \emph{source} denoted $s$ and a \emph{sink} denoted $t$. A path in $G$ from $s$ to $t$ is called an \emph{$st$-path} and it is denoted $s\leadsto t$. The in-degree of $s$ and the out-degree of $t$ are always 0. If $v$ is a vertex in $G$, an st-path passing through $v$ is denoted $s\overset{v}{\leadsto} t$.

A \emph{flow} in $G$ is a function $f:E(G)\to \mathbb N$ that satisfies the following three constraints: (i) a \emph{capacity constraint on the edges} $0\leq f(u,v) \leq c(u,v)$ for all $(u,v)\in E(G)$, (ii) a \emph{capacity constraint on the vertices} $\sum_{(u,v)\in E(G)} f(u,v)\leq c(v)$ for all $v\in V(G)$, and (iii) a \emph{flow conservation constraint} $\sum_{v\in V(G)} f(v,u) = \sum_{v\in V(G)} f(u,v)$ for all $u\in V(G)\setminus\{s,t\}$. If $e$ is an edge of $G$ we call $f(e)$ the flow through $e$, and the flow through an st-path is defined as $f(s\leadsto t) =\min_{e\in E(G)} \{f(e)\}$. The \emph{value} of a flow $f$ is defined as $|f|=\sum_{v\in V(G)} f(s,v)$. The \emph{maximum-flow problem} is thus the problem of finding a flow of maximum value.

\section{Proof of Theorem \ref{the:ptromino}}\label{sec:main-theorem}
Let $R$ be a connected region with $n$ cells and let $P$ be a collection of cells from $R$ that contain pegs, where $|P|=n/3$. We assume that $n$ is a multiple of 3, since, otherwise the problem is trivially  unsolvable. The following procedure is an algorithm that takes an instance $(R,P)$ from $\textsc{P-Tromino}_{n,n/3}$ and transforms it to an instance $(G,n/3)$ of $\textsc{MaxFlow}_{n/3}$.

\begin{enumerate}
\item Color each cell $R_{(a,b)}\in R$ according to its row number. A cell in an even row is colored black, and a cell in an odd row is colored white; see Fig. \ref{fig:coloring}.
\item Partition the cells of $R \setminus P  $ into two sets $B$ and $W$. We define $B$ to be the set of all cells of $ R$ that are colored black and do not contain a peg, that is, $B=\{R_{(a,b)}\in R\setminus P:\> a \text{ is even}\}$. We define $W$ to be the set of all cells of $R$ that are colored white and do not contain a peg, that is, $W=\{R_{(a,b)}\in R \setminus P:\> a \text{ is odd}\}$.
\item Construct a flow network $G=(V,E)$ as follows.
    \begin{enumerate}
    \item Let $s$ be a source vertex and $t$ be a sink vertex.
    \item Each cell $R_{(a,b)}$ in $R$ is a vertex $(a,b)$ in $V$.
    \item Add an edge from $s$ to each vertex $(a,b)$ whose corresponding cell is colored black, that is, $R_{(a,b)}\in B$.
    \item Add an edge from each vertex $(a,b)$ whose corresponding cell is colored white to $t$.
    \item For each black cell $R_{(a,b)}\in B$ and each cell with a peg $R_{(c,d)}\in P$, there is an edge $((a,b),(c,d))$ in $E$ if and only if $R_{(a,b)}$ and $R_{(c,d)}$ are adjacent in the region $R$; see for example Fig. \ref{fig:peg-vertex}.
    \item For each cell with a peg $R_{(c,d)}\in P$ and for each white cell $R_{(a,b)}\in W$, there is and edge $((c,d),(a,b))$ in $E$ if and only if $R_{(c,d)}$ and $R_{(a,b)}$ are adjacent in the region $R$.
    \item All edges in items (e) and (f) have a capacity of 1, all edges coming from $s$ and all edges going to $t$ have capacity $1$, and all remaining missing edges have a capacity of 0.
    \item Each vertex $(a,b)\in V(G)\setminus \{s,t\}$ has $c((a,b))=1$ and $c(s)=c(t)=\infty$.
    \end{enumerate}
\end{enumerate}

\begin{figure}[t]
    \centering
    \includegraphics[scale=1]{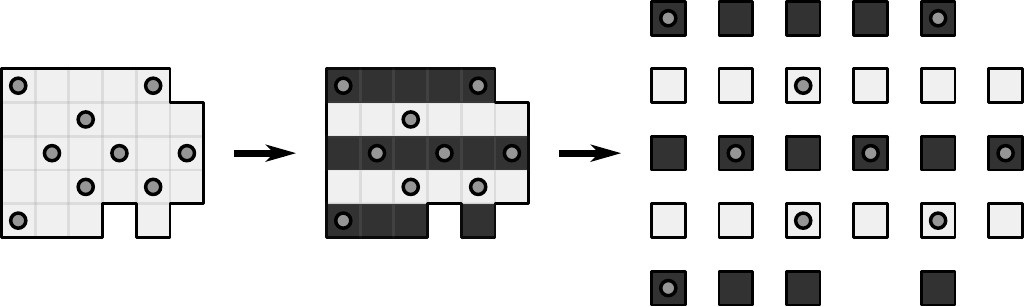}
    \caption{Coloring of the cells with respect the row number. We start counting from row 0 which is the bottom row of the region.}
    \label{fig:coloring}
\end{figure}

\begin{figure}[t]
    \centering
    \includegraphics[scale=1.5]{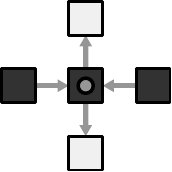}
    \caption{Each cell with a peg has at most 2 incoming edges from black cells and at most 2 outgoing edges to white cells. The adjacent cells on the sides of a cell with a peg in the region $R$ are always of the same color, for example, a black cell with a peg always has black horizontal neighbors and white vertical neighbors.}
    \label{fig:peg-vertex}
\end{figure}

In this work we refer to the graph $G$ obtained from the procedure above as the \emph{region network} of $R$, see for example Fig. \ref{fig:region-network}. It is easy to see that any region network has $n + 2$ vertices and at most $2n/3+4n/3$ edges, and can be constructed in $O(n)$ time. Also note that a region network is a dual graph with two new vertices, some edges removed, and direction given to the remaining edges.


\begin{figure}[t]
    \centering
    \begin{subfigure}[b]{.49\linewidth}
        \centering
        \includegraphics[scale=1]{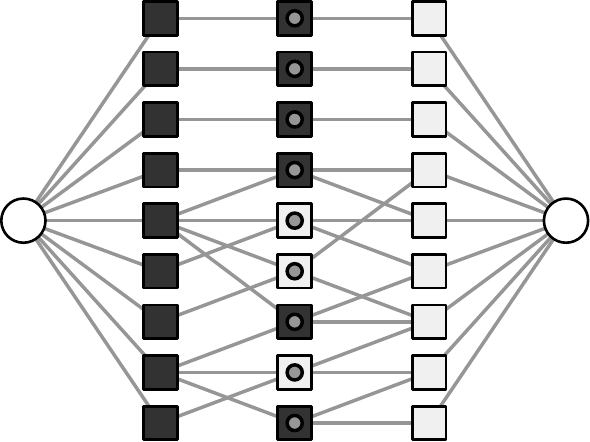}
    \end{subfigure}
    \begin{subfigure}[b]{.49\linewidth}
        \centering
        \raisebox{.685cm}{\includegraphics[scale=1]{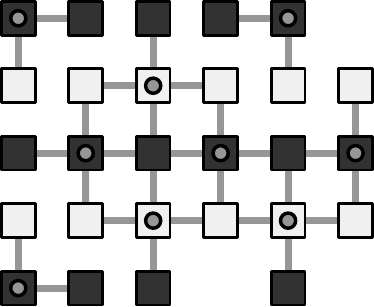}}
    \end{subfigure}
    \caption{Region network of the region on the right. Note that there are no edges (incoming or outgoing) between cells with no pegs. Similarly, there are no edges between cells with pegs.}
    \label{fig:region-network}
\end{figure}

\begin{proposition}\label{pro:reduction}
The value of a maximum flow in $G$ equals $|R|/3$ if and only there is a p-cover of $R$.
\end{proposition}

Proposition \ref{pro:reduction} follows from lemmas \ref{lem:onlyif} and \ref{lem:if} below. First, however, we show a technical lemma concerning the uniqueness of st-paths in the region network.

\begin{lemma}\label{lem:unique-flow}
Let $f$ be a flow in $G$ with $|f|=n/3$. For any $R_{(a,b)}$ there exists an unique path $s\overset{(a,b)}{\leadsto} t$ with $f(s\overset{(a,b)}{\leadsto} t)=1$.
\end{lemma}
\begin{proof}
The flow passing through $(a,b)$ is upper bounded by 1, either because $c((a,b))=1$ if $R_{(a,b)} \in P$ or because the in-degree or out-degree is 1 if $R_{(a,b)} \in B$ or $R_{(a,b)} \in W$, respectively. Therefore,  if there is flow passing through $(a,b)$ we have that there is an unique path $s\overset{(a,b)}{\leadsto} t$ with $f(s\overset{(a,b)}{\leadsto} t)=1$. In order to finish the proof we need to show that $f(s\overset{(a,b)}{\leadsto} t)=1$ for all $R_{(a,b)}$. Indeed, suppose otherwise that there exists $(a,b)$ such that for any path $s\overset{(a,b)}{\leadsto} t$ the flow through $(a,b)$ is 0.  Then, if $R_{(a,b)} \in B$ (respectively $P$, $W$), the flow through the edge-cut $(\{s\},\{s\}^c)$ (respectively $(\{s\}\cup B,  P\cup W\cup \{t\})$, $(\{s\}\cup B\cup P,  W\cup \{t\})$) will be of value $n/3-1$, which contradicts the fact that the flow has value $n/3$.
\end{proof}

\begin{lemma}\label{lem:onlyif}
If the value of a maximum flow in $G$ equals $n/3$, then the region $R$ with $k=n/3$ pegs has a p-cover.
\end{lemma}
\begin{proof}
Let $f$ be a flow of maximum value in $G$, with $|f|=n/3$. Let $R_{(a,b)}$ be a cell with a peg. By Lemma \ref{lem:unique-flow} we known that there is a unique path $s\overset{(a,b)}{\leadsto} t$ with a flow of 1 going through it. Moreover, from our construction of $G$, a path $s\overset{(a,b)}{\leadsto} t$ always has the form
\begin{equation}\label{eq:path}
s\to (c,d) \to (a,b) \to (c',d') \to t,
\end{equation}
where $R_{(c,d)}$ is a black cell with no peg and $R_{(c',d')}$ is a white cell with no peg. 

Now we can construct a p-cover of $R$ using the following ``p-tromino placement rule'': 
place the tromino $T=R_{(c,d)}\cup R_{(a,b)}\cup R_{(c',d')}$ on top of $R_{(a,b)}$. In Fig. \ref{fig:maxflow} we show a correspondence between a flow in the region network and a tromino cover.

\begin{figure}[t]
    \centering
    \begin{subfigure}[b]{.49\linewidth}
        \includegraphics[scale=1]{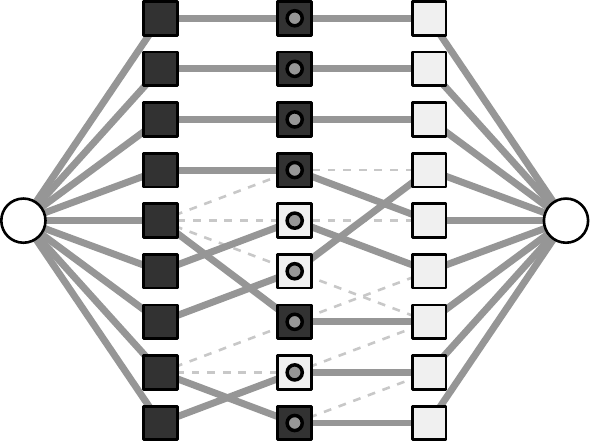}
        \centering
    \end{subfigure}
    \begin{subfigure}[b]{.49\linewidth}
        \raisebox{.685cm}{\includegraphics[scale=1]{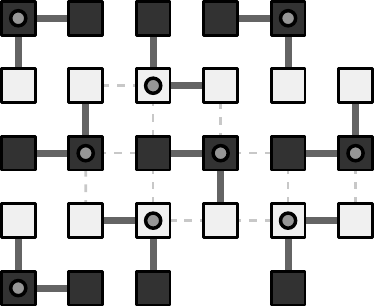}}
        \centering
    \end{subfigure}
    \caption{Maximum flow in a region network on the left. Dashed lines are the edges of the graphs with no flow, and solid lines are the edges with flow. The right subfigure shows the region network over the region.}
    \label{fig:maxflow}
    \centering
\end{figure}

To finish the proof we need to show that our placement of the trominoes is a p-cover of $R$. This follows from Lemma~\ref{lem:unique-flow} since for each cell $R_{(a,b)}$ the vertex $(a,b)$ appears in exactly one st-path.
\end{proof}

\begin{lemma}\label{lem:if}
If the region $R$ with $k=n/3$ pegs has a p-cover, then the value of a maximum flow in $G$ equals $n/3$.
\end{lemma}
\begin{proof}
Let $C$ be a p-cover of size $n/3$. We construct a maximum flow $f$ whose value is $n/3$.

Take any tromino $T\in C$. With no loss of generality, suppose that $T=\trominodr$ and that its corner cell $R_{(a,b)}$ is colored white. Hence, the upper cell $R_{(c,d)}$ is black and its left tip $R_{(c',d')}$ is white. Thus we have a path $s\overset{(a,b)}{\leadsto} t$ in the network graph $G$ of the form $s\to {(c,d)} \to {(a,b)} \to {(c',d')} \to t$. For all edges $e$ of this path we assign a flow $f(e)=1$. We do this for all other trominos in $C$.

First we show that $f$ is a flow with value $|f|=n/3$. Since $C$ is a cover, there are no overlapping trominoes and all cells are covered. This is equivalent to say that each vertex of $G$ that is not the source $s$ or the sink $t$ appears in exactly one st-path $s\leadsto t$ with $f(s\leadsto t)=1$. This implies that $f$ is a flow because each vertex has an incoming and outgoing flow of 1. Furthermore, since we have $n/3$ trominoes in $C$, there are $n/3$ such st-paths, and hence, $|f|=n/3$.

To finish the proof, we show that $|f|=n/3$ is maximum. This follows directly by considering the cut $(\{s\},\{s\}^c)$ consisting on the $n/3$ edges incident with the source $s$, and any flow is bounded by the capacity of any cut.
\end{proof}

With lemmas \ref{lem:onlyif} and \ref{lem:if} we finish our proof of Proposition \ref{pro:reduction}. Now we show that our reduction is parsimonious.
\begin{lemma}
The reduction is parsimonious.
\end{lemma}
\begin{proof}
Let $H$ be our reduction presented in this section and let us fix a region $R$ and a list of pegs $P$. Recall that $H$ is parsimonious if it preserves the number of yes-certificates. That is, the sets
\[\{C:\text{$(R,P)\in \textsc{P-Tromino}_{n,n/3}$  and $C$ is a p-cover of $(R,P)$}\}\]
and
\[\{f:\text{$(H(R,P),n/3)\in \textsc{MaxFlow}_{n/3}$, $f$ is a flow  in $H(R,P)$ with $|f|=n/3$}\}\]
are of the same size.

The existence of a bijective function between both sets above follows from the following argument. By our p-tromino placement rule of Lemma \ref{lem:onlyif}, to any unique st-path in $H(R,P)$ of the form of Eq.(\ref{eq:path}) with flow passing through it,  there is exactly one p-tromino $T\in C$; a different st-path corresponds to a different p-tromino $T'\in C$. Likewise, for any $T\in C$ there is exactly one st-path of the form of Eq.(\ref{eq:path}) with a flow of 1 passing through the vertices corresponding to the cells in $T$.
\end{proof}

\section{Proof of Theorem \ref{the:counting}}\label{sec:counting}
Let us start by stating some trivial facts about matchings in graph theory. (i) The number of perfect matchings of a graph is the product of the number of perfect matchings of its components. (ii) The number of perfect matchings of a cycle is 0 or 2 whenever the  cycle is  odd or even, respectively. (iii) The number of perfect matchings of the union of $k$ cycles is $2^k$ if the cycles are all even and 0 if some of them are odd.
(iv) If a graph has a \emph{leaf} $v$, i.e, a node with an unique incoming edge $(v,w)$, then that edge should be in any perfect matching of the graph.

Let $G = (P \cup B \cup W \cup \{s,t\}, E)$ be the region network that is constructed with our reduction of Section \ref{sec:main-theorem}.
From $G$, let us consider the following two graphs $G'= (P \cup B \cup \{s,t'\}, E')$ and $G''= (P \cup W \cup \{s',t\}, E'')$ given by $E' = E(G-W-\{t\}) \cup P\times\{t'\}$ and $E'' = E(G-B-\{s'\}) \cup \{s'\}\times P$, with distinguished sources and sinks $s,t'$ and $s',t$, respectively. Note that $G$ has a flow  of  value $n/3$  if and only if there exists a flow in $G'$ and a flow in $G''$ both with value $n/3$. Furthermore, the number of flows of value $n/3$ in $G$ equals the number of flows with value $n/3$ in $G'$ and $G''$ multiplied.

\begin{lemma}\label{lem:h}
The number of  flows  in $(G', s, t')$  of value $n/3$  is the same as the number of perfect matchings in $H = G'-s-t'$.
\end{lemma}
\begin{proof}
It suffices to prove that the map that takes a flow $f$ and assigns a set $M$ of edges with flow equals to 1 is a bijection between the set of flows of value $n/3$ and perfect matchings in $H$.
The map is clearly one to one, hence, we only need to check that $M$ is a perfect matching and that any perfect matching is the image of a flow $f$.
In order to check the former we need to prove that there are no adjacent edges in $M$.
Indeed, if for instance, $(x,y)$ and $(x',y)$ are two such edges, then the flow conservation constraint  cannot be satisfied on vertex $y$ because the capacity of edge $(y,t')$ is 1. 
Therefore, the edges in $M$ form a matching in $H$, which is also perfect because it has value $|B|$. 
Finally, let $M$ be a  given a perfect matching on $H$. If we define the flow $f$ to be 1 for the edges in the matching $M$ and for the edges incident with $s$ and $t'$, but 0 otherwise, then  clearly  $f$ is a flow on $(G', s, t')$ with image $M$.
\end{proof}

\renewcommand{\algorithmicrequire}{\textbf{Input:}}
\renewcommand{\algorithmicensure}{\textbf{Output:}}
\begin{algorithm}[t]
    \centering
    \begin{algorithmic}[1]
    \Require graph $H$
    \Ensure number of perfect matchings of $H$
    \State $M \leftarrow  \emptyset$; $H' \leftarrow H$;
    \While{there is a leaf $v$ in $H'$ with unique incoming edge $(w,v)$}
        \State $ M \leftarrow M \cup \{(v,w)\}$;
        \State $ H' \leftarrow  H'-v-w $;
    \EndWhile
    \State Let $H_1,\ldots, H_k$ be the connected components of $H'$, where $k=0$ if $V(H')=\emptyset$;
    \If{there is $i$ such that $H_i$ has an odd cycle or an isolated vertex}
        \State \Return 0
    \Else
        \State \Return $2^k$ 
    \EndIf
    \end{algorithmic}
    \caption{Counting perfect matchings in $H$.}
    \label{alg:counting-matching}
\end{algorithm}

In Algorithm \ref{alg:counting-matching} we present a method for counting perfect matchings on the graph $H$ of Lemma \ref{lem:h}. Let us prove the correctness of Algorithm \ref{alg:counting-matching}.

The {\bf while} loop of lines 2--5 has one  key invariant, namely that  \emph{variable $M$ is always a matching of $H$ and $M$ should be contained in any possible perfect matching of $H$.}
Another invariant is that \emph{the number of vertices of $H'$ in $P$ is the same as the number of vertices of $H'$ in $B$}.
The  {\bf while} loop always finishes because those edges incident to pending vertices must be in any perfect matching.

In lines 3--4 inside the {\bf while} loop, Algorithm \ref{alg:counting-matching} erases those pending vertices together with its unique neighbor, since its incident edge should be on any perfect matching of $H'$. Thus, the new graph $H'$ obtained  contains a perfect matching if and only if $H$ does.

Let us call $\hat H$ the graph $H'$ at the end of the loop. By the condition of line 2 of the {\bf while} loop, the graph $\hat H$ has  no leafs.

\begin{lemma}\label{lem:cycles}
$\hat H$  either has some isolated vertices or is the union of   cycles.
\end{lemma}
\begin{proof}
By the construction of graph $H$, a vertex $v$ of $H$ has degree $d_H(v)\leq 2$ if $v \in P$ and $d_H(v)\leq 4$ if it is in $B$.
The same holds for vertices on $H'$ (and then on $\hat H$), since $H'$ is a subgraph of $H$. 

Let $\hat P$ (resp. $\hat B$) be those vertices of $\hat H$ in $P$ (resp. in $B$), i.e., $\hat P = V(\hat H) \cap V$ and $\hat B = V(\hat H) \cap B$.
Notice that by the second invariant $|\hat P| = |\hat B| $.

We need to prove that if $\hat H$ has no isolated vertices then it is the union of cycles, i.e., all its vertices have degree 2.
 Indeed, since  all vertices have degree greater than 1, all vertices in $\hat P$ should have degree 2. Now
$$
2 |\hat P| = \sum_{v\in \hat P} 2 = \sum_{v\in \hat P} d_{\hat H} (v)  = 
|E(\hat H)| =  \sum_{v\in \hat B} d_{\hat H} (v) \geq \sum_{v\in \hat B} 2 = 2 |\hat B|.
$$
However, $|\hat P| = |\hat B|$, and hence, the inequality should be an equality which happens only if $d_{\hat H} (v) = 2$ for all $v\in \hat B$. Therefore, the vertices of $\hat H$ have degree 2, and $\hat H$ should be the union of some number  of  cycles.
\end{proof}

From Lemma \ref{lem:cycles} it follows that the connected components of line 6 in Algorithm \ref{alg:counting-matching}  are all cycles or have some isolated vertex. Therefore if either $\hat H$ has an isolated vertex or an odd cycle, then there will be no perfect matching; otherwise, $\hat H$ will have $k$ even cycles and, therefore,  $2^k$ perfect matchings. This proves the correctness of Algorithm \ref{alg:counting-matching}.

Algorithm \ref{alg:counting-matching} runs in linear time maintaining the degree of each vertex of $H'$ and a list of the leafs of $H'$. 
Since the degree is upper bounded by 4, computing the degrees takes at most $2\times 4n/3$ time, and in that run it is possible to build a list of leafs.
Given a list of leafs, the condition of the {\bf while} loop uses constant time, while the deletion of vertices in line 5 also takes constant time because of the bound on the degrees.
After line 5 is executed, the degrees are refreshed at those vertices incident with vertex $w$, and if one of them becomes a leaf, it will be added to the list of leafs. All theses operations require constant time.
The {\bf while} loop is executed at most $n/3$  times. The connected components of $\hat H$ are also computed in linear time.
Finally, the number of flows in $G$ is obtained by multiplying the number of matchings in $H$ times the number of matchings in $H''= G''- s'-t$.

\section{Concluding Remarks and Open Problems}\label{sec:conclusions}
In this work we showed how to characterize the tromino tiling problem with pegs using flow networks. We called this characterization the \emph{region network} representation of a region with pegs. Then we showed that perfect matchings in bipartite subgraphs of a region network correspond to p-covers in a region with pegs. Thus, the number of perfect matchings in such bipartite graphs give us the number of p-covers and we showed that counting these perfect matchings can be done in linear-time. All these results present a tight connection between tromino tilings with pegs, maximum flows in flow networks and perfect bipartite matchings, thus giving us a new way to understand tromino tiling problems. 

Below we give a couple of open problems that we believe are challenging and will deepen our understanding of tromino tilings in general.
\begin{enumerate}
    \item In Theorem \ref{the:ptromino} we showed that our reduction holds only when the number of pegs is $n/3$ and $n$ is a multiple of 3. Suppose now that beside p-trominoes, we also have normal (with no hole) trominoes at our disposal. A natural question to ask is what will happen if we have $n/3-k$ pegs and $k$ is given. Intuitively, if $k$ is ``small,'' we can construct in linear-time a region network and run a polynomial-time algorithm for maximum flow to place p-trominoes on pegs. Then, for the squares that remain to be covered we use normal trominoes and, since $k$ is close to 0, we can find a tiling by brute force and thus in constant-time. On the other hand, if $k$ is ``large,'' we have only a few pegs on the region and our brute force approach does not give us an efficient algorithm and in general should be NP-complete. What are the values of $k$ for which the P-Tromino tiling problem is NP-complete or admits a polynomial-time algorithm? We believe the P-Tromino tiling problem is fixed-parameter tractable with $k$ as a parameter.
    \item The P-Tromino tiling problem, as defined in this paper, is given by a region with pegs and trominoes with holes. Now let us suppose that we interchange those roles and this time the trominoes have pegs in their corner cells and the region has holes. Suppose further that we can place a tromino on any hole on the region, even if a tip of the tromino is covering another hole. This version of the problem is similar to a p-tromino tiling but it is not the same, because now the tips of the trominoes can cover other holes on the region. Here we need to change our construction of a region network to allow edges between cells with pegs. In this new version of the problem we can also ask: how does the number of holes on the region affects the computational complexity of the tiling problem? If the number of holes on the region is close to $n$, then the tiling problem is NP-complete. Then, for which number of holes on the region does the problem admits a polynomial-time algorithm?
\end{enumerate}

\bibliographystyle{plain}
\bibliography{tromino-flow}

\end{document}